\documentclass[conference]{IEEEtran}%
\usepackage{amsfonts}
\usepackage{amsmath}
\usepackage{amssymb}
\usepackage{graphicx}
\usepackage{cite}%
\providecommand{\U}[1]{\protect\rule{.1in}{.1in}}
\newtheorem{theorem}{Theorem}

\newtheorem{definition}{Definition}
\newtheorem{example}{Example}

\newtheorem{proposition}{Proposition}
\newtheorem{remark}{Remark}

\begin{document}

\title{Multi-User Privacy: The Gray-Wyner System and Generalized Common Information}
\pubid{~}
\specialpapernotice{~}

%

\author{\authorblockN{Ravi Tandon, Lalitha Sankar, H. Vincent Poor}
\authorblockA{Dept. of Electrical Engineering,\\
Princeton University,
Princeton, NJ 08544.\\
Email: \{rtandon,lalitha,poor\}@princeton.edu\\}}%
%

\maketitle
%

\begin{abstract}%

\footnotetext{The research was supported by the Air Force Office of Scientific
Research MURI Grant FA-$9550$-$09$-$1$-$0643$, by the National Science
Foundation Grants CNS-$09$-$05398$ and CCF-$10$-$16671,$ and by a Fellowship
from the Council on Science and Technology at Princeton University.}The
problem of preserving privacy when a multi-variate source is required to be
revealed partially to multiple users is modeled as a Gray-Wyner source coding
problem with $K$ correlated sources at the encoder and $K$ decoders in which
the $k^{th}$ decoder, $k=1,2,...,K,$ losslessly reconstructs the $k^{th}$
source via a common link of rate $R_{0}$ and a private link of rate $R_{k}$.
The privacy requirement of keeping each decoder oblivious of all sources other
than the one intended for it is introduced via an equivocation constraint
$E_{k}$ at decoder $k$ such that the total equivocation summed over all
decoders $E\geq\Delta$. The set of achievable $(\{R_{k}\}_{k=1}^{K}%
,R_{0},\Delta)$ rates-equivocation $(K+2)$-tuples is completely characterized.
Using this characterization, two different definitions of common information
are presented and are shown to be equivalent. %

\end{abstract}%

\section{Introduction}

Information sources often need to be made accessible to multiple legitimate
users simultaneously. However, not all data from the source should be
accessible to all users. For example, a computer retailer may need to share
the annual revenue of all computers sold with all the vendors but share
vendor-specific sale information only with a particular vendor. Similarly, a
business consulting firm may share general data about a specific market with
all clients associated with that market but share client-specific strategies
with only that client. In both cases, one can view sharing the public (shared
by all) information via a common link and the private information via a
dedicated link. Maximizing the rate over the common link allows the
information source (retailer/consulting firm) to share the most allowed
publicly with all clients; however, the privacy guarantee requires that no
client has access to private data of the other clients. This paper develops an
abstract model and a methodology to study this problem.

\bigskip

We model the problem of revealing partial source information to multiple users
while keeping the data specific to each user private from other users as a
Gray-Wyner source coding problem with $K$ correlated sources at the encoder
and $K$ decoders in which the $k^{th}$ decoder, $k=1,2,...,K,$ losslessly
reconstructs the $k^{th}$ source via a common link of rate $R_{0}$ and a
private link of rate $R_{k}$. We model the privacy requirement of keeping each
decoder oblivious of all sources other than the one intended for it via an
equivocation constraint $E_{k}$ at decoder $k$ such that the total
equivocation summed over all decoders $E\geq\Delta$.

\bigskip

Since privacy is an important aspect of this problem, it is natural to
understand the maximal total equivocation that is achievable if the rate on
the common link is set to the maximum achievable. On the other hand, imposing
the constraint of maximal total equivocation may lead to perhaps a different
limit on the maximal rate on the common link. In this paper, we show that both
requirements, which are formally different definitions, yield the same
formulation for the maximal rate on the common link. In keeping with the
literature, this common rate is defined as the \textit{common information}.

\bigskip

The common information of two correlated random variables has been defined
independently by Wyner \cite{WynerCI} and G\'{a}cs-K\"{o}rner
\cite{GacsKorner}. Wyner's definition of common information as applied to the
two-user Gray-Wyner system (without privacy constraints) is the minimum rate
on the common link such that the total information shared across all three
links (one common and two private) does not exceed the source entropy. On the
other hand, the G\'{a}cs-K\"{o}rner common information is the maximal entropy
of a random variable that two non-interacting terminals can agree upon when
one terminal has access to $X^{n}$ and the other to $Y^{n}$ where $X$ and $Y$
are correlated random variables. For two correlated variables $X$ and $Y$, the
Wyner common information $C_{W}$, the G\'{a}cs-K\"{o}rner common information
$C_{GK}$, and the mutual information of the two variables are related as
$C_{GK}\leq I(X;Y)\leq C_{W}$. Recently, the authors in \cite{ChenCI} have
generalized Wyner's definition of common information to $K$ variables,
henceforth referred to as $B\left(  X_{1},X_{2},\ldots,X_{K}\right)  $ for $K$
correlated variables. While the definition naturally generalizes the two
variable common information, the resulting common information does not satisfy
a non-increasing property with $K$ as expected.

\bigskip

In this paper, we present two different definitions of common information: the
first is the maximal rate on the common link for which the total equivocation
is maximized, and the second is the maximal rate on the common link such that
each user losslessly reconstructs its intended source at its entropy. We show
that both definitions lead to the same formulation for common information
$C\left(  X_{1},X_{2},\ldots,X_{K}\right)  $. We present many properties of
$C\left(  X_{1},X_{2},\ldots,X_{K}\right)  $ and specifically show that
$C\left(  X_{1},X_{2},\ldots,X_{K}\right)  \leq B\left(  X_{1},X_{2}%
,\ldots,X_{K}\right)  $. To the best of our knowledge this is the first
generalization of common information that preserves the non-increasing
property and one whose form can be viewed as a natural generalization of the
G\'{a}cs-K\"{o}rner common information to $K$ variables.

\bigskip

The paper is organized as follows. In Section \ref{sec:system}, we present the
system model. In Section \ref{sec:results}, we present the rate-equivocation
region, develop a formulation for common information in two different ways,
and present key properties. In Section \ref{sec:compare}, we compare our
formulation with the $K$-variable generalization of Wyner's common information
in \cite{ChenCI} and illustrate with examples. We conclude in Section
\ref{sec:conclusion}.

\section{\label{sec:system}System Model}

We consider the following source network. A centralized encoder observes $K$
discrete, memoryless correlated sources, $\{X_{k}^{n}\}_{k=1}^{K}$ and is
interested in communicating source $X_{k}$ to decoder $k$ in a lossless
manner. The resources available at the encoder comprise two types of noiseless
rate-limited links. There are $K$ links of finite rate from the encoder to
each of the $K$ decoders and there is a common link of finite rate to all
decoders. Figure \ref{figuremodel} shows the source broadcasting network in consideration.%

\begin{figure}
[ptb]
\begin{center}
\includegraphics[
height=1.7841in,
width=3.007in
]%
{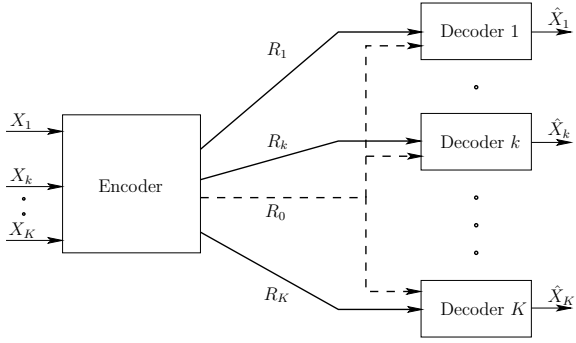}%
\caption{The generalized Gray-Wyner source network.}%
\label{figuremodel}%
\end{center}
\end{figure}

An $(n,\{M_{k}\}_{k=1}^{K},M_{0})$ code for this model is defined by $(K+1)$
encoding functions described as%
\begin{align}
f_{0} &  :\mathcal{X}_{1}^{n}\times\ldots\mathcal{X}_{K}^{n}\rightarrow
\{1,\ldots,M_{0}\},\\
f_{k} &  :\mathcal{X}_{1}^{n}\times\ldots\mathcal{X}_{K}^{n}\rightarrow
\{1,\ldots,M_{k}\},\quad k=1,\ldots,K,
\end{align}
and $K$ decoding functions,
\[
g_{k}:\{1,\ldots,M_{0}\}\times\{1,\ldots,M_{k}\}\rightarrow\mathcal{X}_{k}%
^{n},\quad k=1,\ldots,K.
\]
We define the probability of error at decoder $k$ as
\[
P_{e,k}=\mbox{Pr}(X_{k}^{n}\neq g_{k}(f_{0}(\overline{X}^{n}),f_{k}%
(\overline{X}^{n}))),
\]
where $\overline{X}^{n}\triangleq\{X_{k}^{n}\}_{k=1}^{K}$. We define the
equivocation at decoder $k$ as
\[
E_{k}=\frac{1}{n}H(\overline{X}^{n}\setminus X_{k}^{n}|f_{0}(\overline{X}%
^{n}),f_{k}(\overline{X}^{n})),
\]
and the total equivocation as $E=\sum_{k=1}^{K}E_{k}$.

\begin{remark}
Informally, $E_{k}$ captures the average uncertainty, and hence privacy
achievable, about the remaining $(K-1)$ unintended sources at decoder $k$. 
\end{remark}

An $(\{R_{k}\}_{k=1}^{K},R_{0},\Delta)$ rate-equivocation $(K+2)$-tuple is
achievable for the source network if there exists an $(n,\{M_{k}\}_{k=1}%
^{K},M_{0})$ code such that,
\begin{align}
M_{0}  &  \leq2^{nR_{0}},\\
M_{k}  &  \leq2^{nR_{k}},\quad k=1,\ldots,K\\
P_{e,k}  &  \leq\epsilon_{k},\quad k=1,\ldots,K\\
E  &  \geq\Delta-\epsilon.
\end{align}
We denote by $\mathcal{R}$ the region of all achievable $(\{R_{k}\}_{k=1}%
^{K},R_{0},\Delta)$ rate-equivocation $(K+2)$-tuples.

\section{\label{sec:results}Main Contributions}

\subsection{Rate-Equivocation Region}

We state our first result in the following theorem. The proof is presented in
the appendix.

\begin{theorem}
\label{Th1}The region $\mathcal{R}$ of achievable rates-equivocation $\left(
K+2\right)  $-tuples for the source network shown in Figure \ref{figuremodel}
is the union of all $(k+2)$-tuples $(\{R_{k}\}_{k=1}^{K},R_{0},\Delta)$ that
satisfy%
\begin{align}
R_{0} &  \geq I(X_{1},X_{2},\ldots,X_{K};W),\label{R0Bounds}\\
R_{k} &  \geq H(X_{k}|W),\text{ \ \ }k=1,2,\ldots,K,\label{Rkbounds}\\
\Delta &  \leq%
{\textstyle\sum\limits_{k=1}^{K}}
H\left(  \overline{X}|W,X_{k}\right)  \label{DeltaBounds}%
\end{align}
where the union is over all auxiliary random variables $W$ arbitrarily
correlated with $(X_{1},\,X_{2},\ldots,X_{K})$, and where $\overline{X}%
\equiv\left(  X_{1},X_{2},\ldots,X_{K}\right)  $. 
\end{theorem}

\begin{remark}
The rate region $\mathcal{R}_{G-W}$ of the Gray-Wyner network without
additional equivocation constraints is the region of $\left(  K+1\right)  $
rate tuples that satisfy (\ref{R0Bounds}) and (\ref{Rkbounds}).
\end{remark}

\subsection{Common Information of $K$ Correlated Variables}

We now present two definitions for the common information of $K$ correlated
random variables.

\begin{definition}
\label{defn1} The common information of $K$ correlated random variables,
$C_{1}$, is the maximal value of $R_{0}$, such that $(\{R_{k}\}_{k=1}%
^{K},R_{0},\Delta_{\mbox{max}})\in\mathcal{R}$, where
\[
\Delta_{\mbox{max}}\triangleq\sum_{k=1}^{K}H(\overline{X}|X_{k}).
\]

\end{definition}

\begin{definition}
\label{defn2} The common information of $K$ correlated random variables,
$C_{2}$, is the maximal value of $R_{0}$, such that $(\{H(X_{k})-R_{0}%
\}_{k=1}^{K},R_{0})\in\mathcal{R}_{G-W}$.
\end{definition}

We next state our second result.

\begin{theorem}
\label{theorem2}$C_{1}$ and $C_{2}$ are related as follows:
\begin{equation}
C_{1}=C_{2}=\max\limits_{W-X_{k}-\bar{X}\backslash X_{k},k=1,2,\ldots
,K}I\left(  X_{1}X_{2}\ldots X_{K};W\right)  .
\end{equation}

\end{theorem}

\begin{proof}
From Definition \ref{defn1}, the achievable equivocation $E$ must satisfy
\[
E\geq\Delta_{\mbox{max}}=\sum_{k=1}^{K}H(\overline{X}|X_{k})
\]
On the other hand, any achievable $(\{R_{k}\}_{k=1}^{K},R_{0},E)\in
\mathcal{R}$ also satisfies
\[
E\leq\sum_{k=1}^{K}H(\overline{X}|W,X_{k}).
\]
We therefore, have the following constraint:
\[
\sum_{k=1}^{K}H(\overline{X}|W,X_{k})\geq\sum_{k=1}^{K}H(\overline{X}|X_{k})
\]
which is equivalent to the following $K$ constraints:%
\begin{equation}
I(\overline{X}\setminus X_{k};W|X_{k})=0,\quad k=1,\ldots,K.\label{cons1}%
\end{equation}
Therefore, from Definition \ref{defn1}, $C_{1}$ is equal to the maximal
$R_{0}$ subject to (\ref{cons1}), which implies that%
\[
C_{1}=\max_{W-X_{k}-\overline{X}\setminus X_{k},k=1,\ldots,K}I(X_{1}%
,\ldots,X_{K};W).
\]

From Definition \ref{defn2}, $C_{2}$ is defined as the maximal $R_{0}$ such
that $R_{k}+R_{0}=H(X_{k})$, for $k=1,\ldots,K$, and $(\{R_{k}\}_{k=1}%
^{K},R_{0})\in\mathcal{R}_{G-W}$. We therefore have the following constraints
for $k=1,\ldots,K$:
\begin{align}
H(X_{k}) &  =R_{k}+R_{0}\\
&  \geq H(X_{k}|W)+I(X_{1},\ldots,X_{K};W).
\end{align}
These constraints are equivalent to
\[
I(\overline{X}\setminus X_{k};W|X_{k})=0,\quad k=1,\ldots,K.
\]
Therefore, $C_{2}$ can be written as follows:
\[
C_{2}=\max_{W-X_{k}-\overline{X}\setminus X_{k},k=1,\ldots,K}I(X_{1}%
,\ldots,X_{K};W).
\]

\end{proof}

\subsection{Common Information:\ Properties}

We will now develop some properties of common information of $K$ correlated
random variables defined in Theorem \ref{theorem2}.

\begin{proposition}
\label{prop1}The common information of $K$ random variables, $C\left(
X_{1},X_{2},\ldots,X_{K}\right)  $, is monotonically decreasing in $K$.
\end{proposition}

\begin{proof}
Consider an arbitrary $W$ satisfying the Markov chain relationship
\begin{equation}
W-X_{k}-\overline{X}\setminus X_{k},\quad k=1,\ldots,K.\label{constr}%
\end{equation}
First consider the following sequence of inequalities:
\begin{align}
&  I(X_{1},\ldots,X_{K-1},X_{K};W)\nonumber\\
&  =I(X_{1},\ldots,X_{K-1};W)+I(X_{K};W|X_{1},\ldots,X_{K-1})\\
&  \leq I(X_{1},\ldots,X_{K-1};W)+I(X_{2},\ldots,X_{K};W|X_{1})\\
&  =I(X_{1},\ldots,X_{K-1};W)\label{equalmain}%
\end{align}
where (\ref{equalmain}) follows from the Markov chain relationship
$W-X_{1}-(X_{2},\ldots,X_{K})$. Now consider the following sequence of
inequalities:
\begin{align}
&  C(X_{1},\ldots,X_{K})\nonumber\\
&  =\max_{W-X_{k}-\overline{X}\setminus X_{k},\quad k=1,\ldots,K}%
I(X_{1},\ldots,X_{K};W)\\
&  \leq\max_{W-X_{k}-\overline{X}\setminus X_{k},\quad k=1,\ldots,K}%
I(X_{1},\ldots,X_{K-1};W)\label{stepa}\\
&  \leq\max_{W-X_{k}-\overline{X}\setminus(X_{k},X_{K}),\quad k=1,\ldots
,(K-1)}I(X_{1},\ldots,X_{K-1};W)\label{stepb}\\
&  =C(X_{1},\ldots,X_{K-1})\label{stepc}%
\end{align}
where (\ref{stepa}) follows from (\ref{equalmain}) and (\ref{stepb}) follows
from the fact that the Markov chain relationship $W-X_{k}-\overline
{X}\setminus X_{k}$ implies the Markov chain relationship $W-X_{k}%
-\overline{X}\setminus(X_{k},X_{K})$. Since the random variable $X_{K}$ could
be chosen arbitrarily from the set $(X_{1},\ldots,X_{K})$, (\ref{stepc}) shows
that the common information is monotonically decreasing in $K$.
\end{proof}

\begin{proposition}
\label{prop2}$C\left(  X_{1},\text{ }X_{2},\text{ }\ldots,\text{ }%
X_{K}\right)  $ is upper bounded as
\begin{equation}
C\left(  X_{1},X_{2},\ldots,X_{K}\right)  \leq\min_{i\not =j,i,j=1,2,\ldots
,K}I\left(  X_{i};X_{j}\right)  .
\end{equation}

\end{proposition}

\begin{proof}
We consider an arbitrary $W$ satisfying (\ref{constr}), and upper bound the
following mutual information:
\begin{align}
I(X_{1},\ldots,X_{K};W) &  =I(X_{i};W)+I(\overline{X}\setminus X_{i}%
;W|X_{i})\\
&  =I(X_{i};W)\label{steppa}\\
&  \leq I(X_{i};X_{j},W)\\
&  =I(X_{i};X_{j})+I(X_{i};W|X_{j})\\
&  =I(X_{i};X_{j})\label{steppb}%
\end{align}
where (\ref{steppa}) follows from the Markov chain condition $W-X_{i}%
-\overline{X}\setminus X_{i}$, and (\ref{steppb}) follows from the Markov
chain condition $W-X_{j}-X_{i}$. The choice of $(i,j)$ was arbitrary, and
therefore, the common information is upper bounded by the minimum of pairwise
mutual information among all pairs, i.e.,
\[
C(X_{1},\ldots,X_{K})\leq\min_{i\neq j}I(X_{i};X_{j}).
\]

\end{proof}

\section{\label{sec:compare}Comparison and Examples}

In \cite{WynerCI} Wyner defines the common information of two correlated
random variables $(X_{1},X_{2})$ as
\[
B(X_{1},X_{2})=\inf_{X_{1}\rightarrow W\rightarrow X_{2}}I(X_{1},X_{2};W).
\]
One interpretation of this common information can be obtained from the
Gray-Wyner source network. The common information $B(X_{1},X_{2})$ of two
random variables is given as the smallest value of $R_{0}$ such that
$(R_{1},R_{2},R_{0})\in\mathcal{R}_{G-W}$ and $R_{0}+R_{1}+R_{2}\leq
H(X_{1},X_{2})$. Recently, this notion of common information was generalized
to $K$ correlated random variables in \cite{ChenCI}. The common information,
$B(X_{1},\ldots,X_{K})$, of $K$ correlated random variables, as defined in
\cite{ChenCI}, is given by smallest value of $R_{0}$ such that $(\{R_{k}%
\}_{k=1}^{K},R_{0})\in\mathcal{R}_{G-W}$ and $R_{0}+\sum_{i=1}^{K}R_{k}\leq
H(X_{1},\ldots,X_{K})$. The common information $B(X_{1},\ldots,X_{K})$ is
given as
\[
B(X_{1},\ldots,X_{K})=\inf I(X_{1},\ldots,X_{K};W)
\]
where the infimum is over all distributions $p(w,x_{1},\ldots,x_{K})$ that
satisfy%
\begin{align}
\sum_{w\in\mathcal{W}}p(w,x_{1},\ldots,x_{K}) &  =p(x_{1},\ldots
,x_{K})\label{constrB1}\\
p(x_{1},\ldots,x_{K}|w) &  =\prod_{k=1}^{K}p(x_{k}|w).\label{constrB2}%
\end{align}
It was shown in \cite{ChenCI} that $B(X_{1},\ldots,X_{K})$ is monotonically
increasing in $K$. We believe that any intuitively satisfactory measure of
common information should satisfy the property that the common information
should decrease as the number of random variables increases. In Proposition
\ref{prop1}, we showed that our measure of common information indeed satisfies
this property.

We next prove a property of $B(X_{1},\ldots,X_{K})$ that helps us in comparing
it with our common information $C(X_{1},\ldots,X_{K})$.

\begin{proposition}
\label{prop3}$B\left(  X_{1},X_{2},\ldots,X_{K}\right)  $ is lower bounded as
follows:
\begin{equation}
\max_{i\not =j}I\left(  X_{i};X_{j}\right)  \leq B\left(  X_{1},X_{2}%
,\ldots,X_{K}\right)  .
\end{equation}

\end{proposition}

\begin{proof}
To prove Proposition \ref{prop3}, consider an arbitrary $W$ satisfying the
constraints (\ref{constrB1})-(\ref{constrB2}) and the following sequence of
inequalities:
\begin{align}
I(X_{1},\ldots,X_{K};W) &  \geq I(X_{i};W)\\
&  \geq I(X_{i};X_{j})\label{bstep}%
\end{align}
where (\ref{bstep}) follows from the Markov chain relationship $X_{i}-W-X_{j}%
$, and from the data processing inequality. In arriving at (\ref{bstep}), the
choice of $(i,j)$ was arbitrary, and therefore we can maximize over all pairs
$(i,j)$ such that $i\neq j$ to get the best possible lower bound in this manner.
\end{proof}

Using Propositions \ref{prop2} and \ref{prop3}, we have the following:%
\begin{align}
C(X_{1},\ldots,X_{K})  &  \leq\min_{i\neq j}I(X_{i};X_{j})\nonumber\\
&  \leq\max_{i\neq j}I(X_{i};X_{j})\leq B(X_{1},\ldots,X_{K}).
\label{mainprop}%
\end{align}
We will now give two examples to illustrate the usefulness of our definition
$C(X_{1},\ldots,X_{K})$ over $B(X_{1},\ldots,X_{K})$.

\begin{example}
Consider $K=3$ random variables $(X_{1},X_{2},X_{3})$ such that $X_{1}%
\sim\mbox{Ber}(1/2)$, $X_{2}=X_{1}\oplus N$, where $N\sim\mbox{Ber}(\delta)$
and $X_{3}$ is independent of $(X_{1},X_{2})$. Since $X_{3}$ is independent of
$(X_{1},X_{2})$, these sources have nothing in common and we should expect the
`common information' to be zero. Note that for these sources, $\min_{i\neq
j}I(X_{i};X_{j})=0$, whereas $\max_{i\neq j}I(X_{i};X_{j})=1-h(\delta)$.
Therefore, from (\ref{mainprop}), we have
\[
0\leq C(X_{1},X_{2},X_{3})\leq0\leq1-h(\delta)\leq B(X_{1},X_{2},X_{3}),
\]
which implies that $C(X_{1},X_{2},X_{3})=0$, whereas $B(X_{1},X_{2},X_{3})>0$
for any $\delta\in(0,1/2)$.
\end{example}

\begin{example}
Consider $K=3$ random variables $(X_{1},X_{2},X_{3})$ such that $X_{1}%
=(X_{0},X_{1p})$, $X_{2}=(X_{0},X_{2p})$ and $X_{3}=(X_{0},X_{3p})$, where
$(X_{0},X_{1p},X_{2p},X_{3p})$ are all mutually independent. Since $X_{0}$
appears to be the only common part in all three sources, we should expect the
`common information' to be equal to the entropy of $X_{0}$. Note that for
these sources, $\min_{i\neq j}I(X_{i};X_{j})=\max_{i\neq j}I(X_{i}%
;X_{j})=H(X_{0})$. Therefore, from (\ref{mainprop}), we have
\[
0\leq C(X_{1},X_{2},X_{3})\leq H(X_{0})\leq B(X_{1},X_{2},X_{3}),
\]
It is straightforward to show that for these sources,
\[
C(X_{1},X_{2},X_{3})=B(X_{1},X_{2},X_{3})=H(X_{0}).
\]

\end{example}

Inspired by the above example, we show the following interesting property that
in some sense relates $C(X_{1},\ldots,X_{K})$ to $B(X_{1},\ldots,X_{K})$.

\begin{proposition}
\label{propo4}For a set of sources $X_{1},X_{2},\ldots,X_{K}$ that satisfy%
\begin{equation}
\min_{i\neq j}I(X_{i};X_{j})=\max_{i\not =j}I(X_{i};X_{j}),\label{equalcons}%
\end{equation}
we have%
\begin{equation}
C\left(  X_{1},X_{2},\ldots,X_{K}\right)  =\min_{i\neq j}I(X_{i};X_{j})
\end{equation}%
\begin{equation}%
\begin{array}
[c]{cc}%
\text{if } & B\left(  X_{1},X_{2},\ldots,X_{K}\right)  =\max_{i\not =j}%
I(X_{i};X_{j}).
\end{array}
\end{equation}

\begin{proof}
The constraint (\ref{equalcons}) implies that the mutual information
$I(X_{i};X_{j})$ is the \emph{same} for all $i,j\in\{1,\ldots,K\}$, $i\neq j$.
Let us start with a $W^{\ast}$ that satisfies the infimization constraints for
$B(X_{1},\ldots,X_{K})$ and yields
\begin{align}
B(X_{1},\ldots,X_{K}) &  =\max_{i\neq j}I(X_{i};X_{j})\\
&  =I(X_{i_{0}};X_{j_{0}}),
\end{align}
for some $i_{0}\neq j_{0}$. For this $W^{\ast}$, we have
\begin{align}
I(X_{i_{0}};X_{j_{0}}) &  =\max_{i\neq j}I(X_{i};X_{j})\\
&  =I(X_{1},\ldots,X_{K};W^{\ast})\\
&  =I(X_{i_{0}};W^{\ast})+I(\overline{X}\setminus X_{i_{0}};W^{\ast}|X_{i_{0}%
})\\
&  \geq I(X_{i_{0}};X_{j_{0}})+I(\overline{X}\setminus X_{i_{0}};W^{\ast
}|X_{i_{0}})\label{equu}%
\end{align}
where (\ref{equu}) follows from the fact that $W^{\ast}$ satisfies the Markov
relationship $X_{i_{0}}-W^{\ast}-X_{j_{0}}$, for all $i_{0}\neq j_{0}$. In the
derivation of (\ref{equu}), $i_{0}$ can be chosen arbitrarily due to
(\ref{equalcons}). Therefore, (\ref{equu}) implies that this $W^{\ast}$ also
satisfies
\[
I(\overline{X}\setminus X_{i};W^{\ast}|X_{i})=0
\]
for all $i=1,\ldots,K$. This in turn implies that $W^{\ast}$ serves as a valid
choice in the maximization for evaluation of $C(X_{1},\ldots,X_{K})$.
Therefore, we obtain the following lower bound for $C(X_{1},\ldots,X_{K})$:
\begin{align}
C(X_{1},\ldots,X_{K}) &  =\max_{W-X_{k}-\overline{X}\setminus X_{k}%
,k=1,\ldots,K}I(X_{1},\ldots,X_{K};W)\\
&  \geq I(X_{1},\ldots,X_{K};W^{\ast})\\
&  =\max_{i\neq j}I(X_{i};X_{j})\\
&  =\min_{i\neq j}I(X_{i};X_{j}).
\end{align}
Hence, from Proposition \ref{prop1}, it now follows that if $B(X_{1}%
,\ldots,X_{K})=\max_{i\neq j}I(X_{i};X_{j})$, then $C(X_{1},\ldots,X_{K}%
)=\min_{i\neq j}I(X_{i};X_{j})$. We remark here that a similar property has
been shown for $K=2$ by Ahlswede and K\"{o}rner in \cite{AhlswedeCI}.
\end{proof}
\end{proposition}

\section{\label{sec:conclusion}Concluding Remarks}

We have abstracted the problem of privacy in a setting where a source
interacts with multiple users via the Gray-Wyner source coding problem with
additional equivocation constraints at each user and a total equivocation
constraint. In addition to developing the rate-equivocation region, we have
introduced two definitions of common information of $K$ correlated variables
and shown them both to have a form that can be viewed as a $K$-user
generalization of the G\'{a}cs-K\"{o}rner common information (see also
\cite{AhlswedeCI}).

\section{Appendix: Proof of Theorem \ref{Th1}}

The converse follows by minor modifications of the converse proof for the
unconstrained Gray-Wyner problem \cite{GrayWyner} and is therefore omitted. We
now outline the proof of achievability for Theorem \ref{Th1}.

Codebook generation: Fix an input distribution $p(w|x_{1},\ldots,x_{K})$.
Generate $2^{nI(X_{1},\ldots,X_{K};W)}$ sequences according to the
distribution $\prod_{t=1}^{n}p(w_{t})$, and index these sequences as
$w^{n}(i)$, for $i=1,\ldots, 2^{nI(X_{1},\ldots, X_{K};W)}$. Independently and
uniformly bin the $X_{k}^{n}$-sequences in $2^{nH(X_{k}|W)}$ bins, and index
these bins as $b_{k,1},\ldots,b_{k,2^{nH(X_{k}|W)}}$, for $k=1,\ldots,K$.

Encoding scheme: Upon observing the $(x_{1}^{n},\ldots,x_{K}^{n})$ sequences,
the encoder searches for a $w^{n}$ sequence that is jointly typical with these
sequences. Using standard arguments (as in \cite{Cover:book}), it can be shown
that the encoder can succeed in finding one such $w^{n}$ sequence. The encoder
sends the index of the $w^{n}$ sequence on the public link, for which we
require $R_{0}\geq I(X_{1},\ldots,X_{K};W)$. It sends the bin index of the
source sequence $x_{k}^{n}$ on the private link to decoder $k$, for which we
require $R_{k}\geq H(X_{k}|W)$.

Decoding: At decoder $k$, the decoder looks for a unique $x^{n}$ in bin
$b_{k}$ (received from the private link), that is jointly typical with the
$w^{n}$ sequence received from the public link. It can be shown that decoder
$k$ can reconstruct $X_{k}^{n}$ with a vanishingly small probability of error.
We omit the probability of error calculation as it follows from the same
arguments as in \cite{GrayWyner}.

Equivocation: We show that this coding scheme yields the total equivocation
stated in Theorem \ref{Th1}. Let $J_{0}$ denote the encoder output for the
public link and let $J_{k}$ denote the encoder output for the private link to
decoder $k$, for $k=1,\ldots,K$. For $E_{k}$, we have the following sequence
of inequalities:
\begin{align}
E_{k} &  =\frac{1}{n}H(X_{1}^{n},\ldots,X_{k-1}^{n},X_{k+1}^{n},\ldots
,X_{K}^{n}|J_{0},J_{k})\\
&  =\frac{1}{n}H(\overline{X}^{n}\setminus X_{k}^{n}|J_{0},J_{k})\\
&  \geq\frac{1}{n}H(\overline{X}^{n}|J_{0},J_{k})-\frac{1}{n}H(X_{k}^{n}%
|J_{0},J_{k})\\
&  \geq\frac{1}{n}H(\overline{X}^{n}|J_{0},J_{k})-\epsilon_{k,n}%
\label{reconss}\\
&  =\frac{1}{n}H(\overline{X}^{n},J_{0},J_{k})-\frac{1}{n}H(J_{0}%
,J_{k})-\epsilon_{k,n}\\
&  \geq\frac{1}{n}H(\overline{X}^{n})-\frac{1}{n}H(J_{0},J_{k})-\epsilon
_{k,n}\\
&  \geq\frac{1}{n}H(\overline{X}^{n})-\frac{1}{n}H(J_{0})-\frac{1}{n}%
H(J_{k})-\epsilon_{k,n}\\
&  \geq H(X_{1},\ldots,X_{K})-I(X_{1},\ldots,X_{K};W)-H(X_{k}|W)\nonumber\\
&  \quad-\epsilon_{k,n}\label{eqnequivo}\\
&  =H(X_{1},\ldots,X_{K}|W,X_{k})-\epsilon_{k,n}\\
&  =H(\overline{X}|W,X_{k})-\epsilon_{k,n},
\end{align}
where (\ref{reconss}) follows from Fano's inequality, and (\ref{eqnequivo})
follows from the facts that $H(J_{0})\leq\log(|\mathcal{J}_{0}|)=nI(X_{1}%
,\ldots,X_{K};W)$, and $H(J_{k})\leq\log(|\mathcal{J}_{k}|)=nH(X_{K}|W)$, for
$k=1,\ldots,K$. Therefore, we have that
\[
E=\sum_{k=1}^{K}E_{k}\geq\sum_{k=1}^{K}H(\overline{X}|W,X_{k})-\epsilon.
\]
Hence, this coding scheme yields an equivocation of $\Delta=\sum_{k=1}%
^{K}H(\overline{X}|W,X_{k})$.

\bibliographystyle{IEEEtran}
\bibliography{refravi}

\end{document}